\documentclass[conference,times,psfig,twocolumn]{IEEEtran}

\usepackage{cite}

\usepackage[justification=centering]{caption}
\usepackage[pdftex]{graphicx}
\graphicspath{{./images/}}
\DeclareGraphicsExtensions{.png}

\usepackage{amsfonts}
\usepackage[cmex10]{amsmath}
\usepackage{amsthm}
\newtheorem{definition}{Definition}
\newtheorem{theorem}{Theorem}
\newtheorem{lemma}{Lemma}

\usepackage{algorithmic}
\usepackage{algorithm}
\usepackage{setspace}
\let\Algorithm\algorithm
\renewcommand\algorithm[1][]{\Algorithm[#1]\setstretch{1.3}}

\usepackage{mathtools}
\DeclarePairedDelimiter\ceil{\lceil}{\rceil}

\usepackage{balance}

\begin{document}

\title{Hardware Based Projection onto \\ The Parity Polytope and Probability Simplex}

\author{\IEEEauthorblockN{Mitchell Wasson}
\IEEEauthorblockA{University of Toronto\\
Email: m.wasson@mail.utoronto.ca}
\and
\IEEEauthorblockN{Stark C. Draper}
\IEEEauthorblockA{University of Toronto\\
Email: stark.draper@utoronto.ca}
}

\maketitle

\begin{abstract}
This paper is concerned with the adaptation to hardware of methods for Euclidean norm projections onto the parity polytope and probability simplex. We first refine recent efforts to develop efficient methods of projection onto the parity polytope. Our resulting algorithm can be configured to have either average computational complexity $\mathcal{O}\left(d\right)$ or worst case complexity $\mathcal{O}\left(d\log{d}\right)$ on a serial processor where $d$ is the dimension of projection space. We show how to adapt our projection routine to hardware. Our projection method uses a sub-routine that involves another Euclidean projection; onto the probability simplex. We therefore explain how to adapt to hardware a well know simplex projection algorithm. The hardware implementations of both projection algorithms achieve area scalings of $\mathcal{O}(d\left(\log{d}\right)^2)$ at a delay of $\mathcal{O}(\left(\log{d}\right)^2)$. Finally, we present numerical results in which we evaluate the fixed-point accuracy and resource scaling of these algorithms when targeting a modern FPGA.
\end{abstract}

\section{Introduction}

The need for hardware-compatible algorithms that can project onto the parity polytope, and onto the probability simplex, is motivated by linear program (LP) decoding techniques for LDPC codes. LP decoding of error-correcting codes is a relatively new concept that originated in Jon Feldman's Ph.D. \cite{feldman_phd,feldman_journal}. Since then, much work has been done to make this approach more appealing to practical implementation. For example, progress has been made by approximating the LP decoding problem by a solver that has linear complexity in blocklength \cite{vontobel_journal,burshtein_journal}. Additionally, \cite{seigel_adaptive} showed how the complexity and performance of LP decoding can be improved by starting with a simple linear program and adaptively adding constraints. More recently, \cite{draper_admm} applied the ADMM decomposition technique from \cite{boyd_admm} to LP decoding. This decomposition gives rise to a distributed algorithm which relies heavily on parity polytope projections. We will see that projection onto the probability simplex plays a integral role in parity polytope projections.

Hardware implementation of these projection algorithms are needed in order for ADMM decoding to be practical. This is especially true for applications with stringent reliability requirements and high data rates (e.g., optical transport). These are also the applications for which the theoretical guarantees of LP decoding are especially attractive. Therefore, this study is the first step in determining whether or not ADMM is a commercially viable LP decoding algorithm.

First, we review the state of available projection techniques in Section \ref{algoBackSec}. Then we evaluate and choose algorithms that are suitable for hardware implementation in Section \ref{algoAdaptSec}. For projecting onto the parity polytope, we develop a new technique based on previous approaches, but more compatible with hardware. The correctness of this method is proved. Section \ref{hardAdaptSec} is dedicated to the adaptation of the algorithms to hardware. Finally, we implement the selected projection algorithms in Verilog and synthesize them for an Altera Stratix V FPGA. This allows us to investigate fixed-point precision effects as well as resource scaling. Findings are reported in Section \ref{resultSec}.

Because of the spreading use of FPGA and hardware acceleration in computing, we anticipate the developments and experiments in this paper will be of interest to fields other than error-correction coding (e.g., optimization and machine learning). Therefore, we present our results in full generality, outside the specific context of error-correction.

\section{Background}
\label{algoBackSec}
\subsection{Projecting onto the Parity Polytope}
As visualized in Fig. \ref{ppPic}, the parity polytope is defined as the convex hull of all even weight vertices of the unit hyper-cube, i.e.,

\begin{equation*}
\mathbb{PP}_d := \text{conv}\left(\left\{ e \in \left\{ 0,1 \right\}^d : \| e \|_1 \text{ even} \right\}\right).
\end{equation*}

\begin{figure}
	\centering
	\includegraphics[scale=0.55]{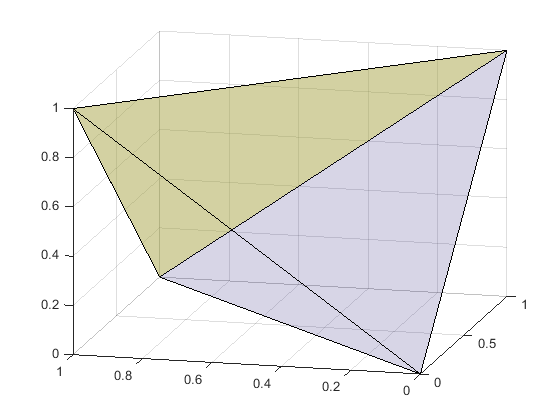}
	\caption{Visualization of $\mathbb{PP}_3$}
	\label{ppPic}
\end{figure}

The parity polytope is a polyhedron. Therefore, as discussed in \cite{draper_admm}, projecting a point $v\in\mathbb{R}^d$ onto $\mathbb{PP}_d$ is a quadratic program,

\begin{equation}
\label{PPproject}
\prod\nolimits_{\mathbb{PP}_d}\left( v \right) = \arg\min_{z \in \mathbb{PP}_d} \| v - z \|_2^2,
\end{equation}

\noindent where (Euclidean) projection onto the set $\mathcal{X}$ is denoted $\prod\nolimits_{\mathcal{X}}\left( \cdot \right)$.

A computationally tractable approach to solving (\ref{PPproject}) was given in \cite{draper_admm}. The approach involves a ``two-slice'' characteristic of points in the parity polytope. These points can be represented as a convex combination of binary vectors with weight $r$ or $r+2$ for some even integer $r$. The authors show that the projection is invariant to the ordering of input vector components. By first applying a sort, finding the two-slice representation of the projection is reduced to an optimization over a piecewise linear function that has $\mathcal{O}\left(d\right)$ complexity. The two-slice approach is limited by the sorting and desorting of vector components, which gives it a computational complexity of $\mathcal{O}\left(d\log{d}\right)$. 

Zhang and Siegel developed a more efficient projection algorithm in \cite{siegel_projection}. Their algorithm first determines the facet of the parity polytope on which the projection lies. They then project onto said facet. Determining the correct parity polytope facet is achieved by employing the authors' ``cut search'' algorithm, fully described in \cite{siegel_cut}. To project a point onto the facet found, the authors use the same optimization over a piecewise linear function described in \cite{draper_admm}. This method is less complex than \cite{draper_admm} because the sort and desort operations in \cite{draper_admm} are replaced by the cut search operation. However, due to this change a sort is still required in solving the piecewise linear optimization problem, resulting in the same order computational complexity $\mathcal{O}\left(d\log{d}\right)$.

Even more recently, another projection algorithm was proposed in \cite{kleijn_projection}. The first improvement is to use a partial sort on input vector components instead of a full sort (as in \cite{draper_admm}) to obtain the same result. Second, the authors show that the projection's two-slice representation can be computed via a projection onto the probability simplex. Both of these operations have expected complexity that is linear in dimension, giving the entire projection algorithm an expected computational complexity of $\mathcal{O}\left(d\right)$. However, worst case complexity is quadratic.

\subsection{Projecting onto the Probability Simplex}
\label{simplexBackSubSec}
The probability simplex $\mathbb{S}_d$ is defined as the set of all positive vectors whose components sum to one. Precisely:

\begin{equation*}
\mathbb{S}_d := \left\{ x \in \mathbb{R}^d : 1^\top x = 1, x_i \geq 0 \right\}.
\end{equation*}

\begin{figure}
	\centering
	\includegraphics[scale=0.5]{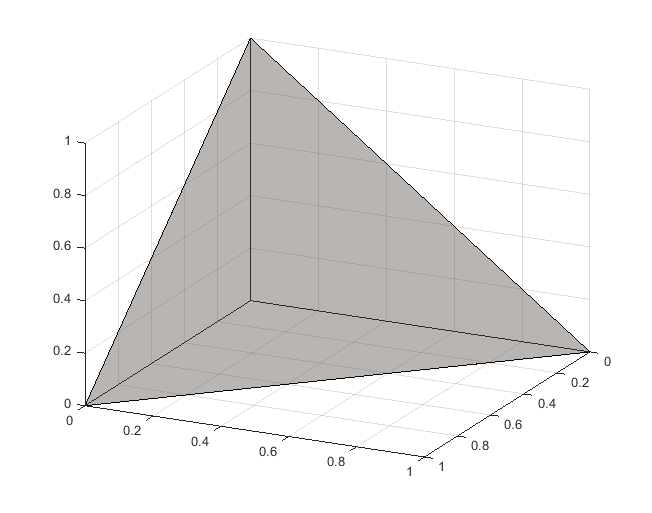}
	\caption{Visualization of $\mathbb{S}_3$}
	\label{simplexPic}
\end{figure}

\noindent Fig. \ref{simplexPic} visualizes the probability simplex in $\mathbb{R}^3$. Similar to the parity polytope, projection onto this set is a quadratic program:

\begin{equation*}
\prod\nolimits_{\mathbb{S}_d}\left( v \right) = \arg\min_{z \in \mathbb{S}_d} \| v - z \|_2^2.
\end{equation*}

An excellent reference for projection onto the probability simplex is found in \cite{duchi_simplex}. The authors first develop a method for simplex projection that uses a sorting of the input components. The complexity is $\mathcal{O}\left(d\log{d}\right)$. Next, the authors develop a linear expected complexity (quadratic worst case complexity) algorithm that performs similar calculations to their first approach in an iterative manner. Additionally, the authors describe how to use projection onto the probability simplex to compute projection onto the $\ell_1$-ball.

\section{Algorithmic Adaptation}
\label{algoAdaptSec}

Before tackling the hardware implementation of the projection algorithms, we evaluate the techniques discussed in Section \ref{algoBackSec} for their hardware compatibility. Criteria we use include computational complexity, delay, parallizability, and input invariance. By input invariance we mean that the same operations are performed regardless of input values.

\subsection{Adapting the Parity Polytope Projection}
The parity polytope projection routines specified in \cite{draper_admm} and \cite{siegel_projection} contain the same overall structure and computational complexity which scales well. The drawback of these algorithms is their iterative nature. This results in high delay, low potential for parallelization, and input dependence.

The projection routine from \cite{kleijn_projection} has good overall structure and excellent computational complexity. However, we believe it is difficult to adapt this algorithm's average complexity to hardware. The technique relies on a partial sort for which we have not found a general hardware technique (except to implement a partial sort via a full sort). Additionally, the simplex projection suggested in \cite{kleijn_projection} is poorly suited for hardware. We will see another routine for the latter task that maps well to hardware.

Each of the algorithms discussed above has its own merits, yet we believe none is suitable for hardware adaptation. Therefore we introduce a novel approach by combining two of the techniques. This new method is similar to that presented in \cite{kleijn_projection}. However, we employ the cut search method used in \cite{siegel_projection} to identify which facet to project onto. We then perform the projection using a probability simplex projection. Doing this avoids the partial sort and replaces it with a guaranteed linear complexity component. Algorithm \ref{pppAlgorithm} describes this method in full. Definition \ref{transformDef} provides a similarity transformation that is used in the proposed algorithm.

\begin{definition}
	\label{transformDef}
	Given a vector $v\in\mathbb{R}^d$ and a vector $f\in\{0,1\}^d$, we define the operator $T_f:\mathbb{R}^d \to \mathbb{R}^d$ as
	
	\begin{equation*}
		T_f\left(v\right)_i = 
		\begin{cases}
		1-v_i&\text{if } f_i = 1\\
		v_i  &\text{if } f_i = 0
		\end{cases}
	\end{equation*}
	
	\noindent where $v = T_f\left(T_f\left(v\right)\right)$. Note that this is very similar to Definition 2 in \cite{kleijn_projection}. 
\end{definition}

\begin{algorithm}
\caption{Hardware Compatible Projection onto $\mathbb{PP}_d$}
\label{pppAlgorithm}
\textbf{Input:} $v\in\mathbb{R}^d$\\
\textbf{Output:} $w \in \mathbb{PP}_d$
\begin{algorithmic}[1]
	\STATE $\hat{v} = \prod\nolimits_{\left[0,1 \right]^d}\left( v \right)$
	
	\STATE $
	f_i= 
	\begin{cases}
	1& \text{if } \hat{v}_i > \frac{1}{2}\\
	0           & \text{otherwise}
	\end{cases}
	$ \hspace{\algorithmicindent} $i=1,\dots,d$
	
	\IF{$1^\top f$ is even}
		\STATE $i^\ast = \arg\min_{\substack{i}} |\frac{1}{2}-\hat{v}_i|$ \label{argminLine}
		\STATE $f_{i^\ast} = 1- f_{i^\ast}$
	\ENDIF

	\STATE $\tilde{v} = T_f\left(v\right)$

	\IF{$ 1^\top \prod\nolimits_{\left[0,1 \right]^d}\left( \tilde{v} \right)  \geq 1$} \label{ppConditionLine}
		\STATE $w = \hat{v}$ \label{ppReturnBoxLine}
	\ELSE
		\STATE $u = \prod\nolimits_{\mathbb{S}_d}\left( \tilde{v} \right)$ \label{ppReturnSimpLineOne}
		\STATE $w = T_f\left( u  \right)$\label{ppReturnSimpLineTwo}
	\ENDIF	
	
\end{algorithmic}
\end{algorithm}

We demonstrate the correctness of the algorithm by partitioning the possible inputs into two cases. The first case is when the unit cube projection of the algorithm input is in the parity polytope. i.e., $\prod\nolimits_{\left[0,1 \right]^d}\left( v \right) \in \mathbb{PP}_d$. The second case covers all other inputs.

As we will see, Theorem \ref{conditionTheorem} (below) provides the test to differentiate between the two cases. This theorem is used on line \ref{ppConditionLine} of Algorithm \ref{pppAlgorithm}. First we introduce a lemma that aids in the development of the theorem.

\begin{lemma}
	\label{swapLemma}
	Given $v\in\mathbb{R}^d$ and $f\in\left\{0,1\right\}^d$,
	\begin{equation}
	\label{swapEquation}
	T_f\left( \prod\nolimits_{\left[0,1 \right]^d} \left( v \right)\right) = \prod\nolimits_{\left[0,1 \right]^d}\left( T_f \left( v \right)\right).
	\end{equation}
\end{lemma}
\begin{proof}
	First note that both $T_f \left( \cdot \right)$ and $\prod\nolimits_{\left[0,1 \right]^d} \left( \cdot \right)$ are component-wise operators. Therefore the lemma can be proved by examining a single component. We now consider the cases of $f_i$. When $f_i=0$, then $T_f(v)_i = v_i$. This gives equality in (\ref{swapEquation}) for all cases of $v_i$. 
	
	When $f_i=1$, we consider three cases for $v_i$. First, if $v_i \leq 0$, the left hand side of (\ref{swapEquation}) becomes $1$. Also, $T_f(v)_i$ will be greater than $1$, making the right hand side $1$. Second, if $v_i \geq 0$, the left hand side becomes $0$. Also, $T_f(v)_i$ will be less than $0$, making the right hand side $0$. Finally, if $0 < v_i < 1$ the left hand side becomes $1 - v_i$. The right hand side also reduces to $1-v_i$ since $1-v_i$ is between $0$ and $1$.
\end{proof}

\begin{theorem}
	\label{conditionTheorem}
	Given $v\in \mathbb{R}^d$ and $f\in\left\{0,1\right\}^d$ as calculated in Algorithm \ref{pppAlgorithm},
	\begin{equation*}
		1^\top \prod\nolimits_{\left[0,1 \right]^d}\left( T_f\left( v \right) \right) \geq 1 \text{ iff } \prod\nolimits_{\left[0,1 \right]^d} \left( v \right) \in \mathbb{PP}_d\text{.}
	\end{equation*}
\end{theorem}
\begin{proof}
	We use Lemma \ref{swapLemma} to transform this inequality into a parity polytope membership test for vectors in the unit cube applied to $\prod\nolimits_{\left[0,1 \right]^d} \left( v \right)$. The parity polytope membership test is that used by the cut search algorithm implementation in \cite{siegel_projection}.
\end{proof}

Since the parity polytope is a subset of the unit cube, we return $\prod\nolimits_{\left[0,1 \right]^d}\left( v \right)$ when $v$ satisfies Theorem \ref{conditionTheorem} (line \ref{ppReturnBoxLine} of Algorithm \ref{pppAlgorithm}). When $v$ does not satisfy Theorem \ref{conditionTheorem}, we find the projection onto the parity polytope via a projection onto the probability simplex (lines \ref{ppReturnSimpLineOne} and \ref{ppReturnSimpLineTwo} of Algorithm \ref{pppAlgorithm}). The correctness of this method is given by Theorem \ref{simplexProjectTheroem} below.

\begin{theorem}
	\label{simplexProjectTheroem}
	Given $v \in \mathbb{R}^d$ such that $\prod\nolimits_{\left[0,1 \right]^d} \left( v \right) \notin \mathbb{PP}_d$ and $f\in\left\{0,1\right\}^d$ as calculated in Algorithm \ref{pppAlgorithm},
	\begin{equation*}
		\prod\nolimits_{\mathbb{PP}_d} \left( v \right) = T_f \left( \prod\nolimits_{\mathbb{S}_d} \left( T_f\left( v \right)  \right) \right)\text{.}
	\end{equation*}	
\end{theorem}
\begin{proof}
	We start with optimization problem (18) presented in \cite{siegel_projection}. We then apply the change of variable $x=T_f \left( y \right)$. This gives us a projection onto the probability simplex.
\end{proof}

This projection algorithm can be configured to have $\mathcal{O}\left(d\right)$ average complexity or $\mathcal{O}\left(d\log{d}\right)$ worst case complexity depending on the implementation of $\prod\nolimits_{\mathbb{S}_d}\left( \cdot \right)$ as outlined in Section \ref{simplexBackSubSec}.

\subsection{Adapting the Probability Simplex Projection}
We considered both methods presented in \cite{duchi_simplex} for hardware based projection onto the probability simplex. One has linear expected complexity, but is highly input dependent and iterative. This makes it a bad candidate for hardware. Because of this, we use the other method from \cite{duchi_simplex} which has $\mathcal{O}\left(d\log{d}\right)$ computational complexity. It is logical that this complexity is necessary since the expected linear time algorithm performs the same computations as the $\mathcal{O}\left(d\log{d}\right)$ algorithm while avoiding, or combining, operations when possible. It is exactly this characteristic that makes the algorithm unsuitable for hardware. The technique we use is described fully by Algorithm \ref{pspAlgorithm}, below. As can be guessed, the descending sort operation sorts the vector components into descending order.

\begin{algorithm}
	\caption{Hardware Compatible Projection onto $\mathbb{S}_d$}
	\label{pspAlgorithm}
	\textbf{Input:} $v\in\mathbb{R}^d$\\
	\textbf{Output:} $w \in \mathbb{S}_d$
	\begin{algorithmic}[1]
		\STATE $\mu = \text{descendingSort}\left(v\right)$
		\STATE $s_i = \frac{1}{i}\left(\sum\limits_{j=1}^{i} \mu_j - 1\right)$ \hspace{\algorithmicindent} $i=1,\dots,d$ \label{prefixSumAl}
		\STATE $\rho = \max\left\{ i \in \left\{ 1,\dots,d \right\} : \mu_i > s_i \right\}$  \label{maxIndex}
		\STATE $w_i = \max\left\{ v_i-s_\rho ,0 \right\}$ \hspace{\algorithmicindent} $\:i=1,\dots,d$
		
	\end{algorithmic}
\end{algorithm}

We have not made any modifications to the statement of this algorithm in \cite{duchi_simplex}.

\section{Hardware Implementation}
\label{hardAdaptSec}

We now discuss the implementation and architectural choices we made
to realize, in hardware, a routine that projects onto the parity
polytope. To measure the efficiency of our implementation we
consider the FPGA area usage of the synthesized algorithm, as well as delay.

We make a comment regarding the fixed-point representation we use in our implementations. Within
certain sub-modules of the projection routine we provide enough
additional dynamic range to avoid overflow. Then, after the
projection is computed, excess bits are truncated to deliver an output
with the same bit width as the projection module input.

\subsection{Implementing the Parity Polytope Projection}

The majority of steps in Algorithm~\ref{pppAlgorithm} comprise
straightforward operations that map directly to parallelizable
hardware constructs with linear complexity and constant delay. There
are two exceptions to this: the simplex projection and the argmin
operation. The simplex projection will be discussed in a subsequent
section. The argmin operation on line~\ref{argminLine} of
Algorithm~\ref{pppAlgorithm} warrants an explanation, as it does not
have constant delay.

\subsubsection{Argmin}
The output of the argmin module is the minimum along with a one-hot
indicator vector that indicates which element of the input vector was
minimal. This calculation is realized using a min-tree where each
node recursively calls two argmin operations. Using the two minima
received, the indicator vector of the higher value is zeroed and
concatenated with the indicator vector of the lower value. This
construction has linear complexity and logarithmic delay.

\subsection{Implementing the Probability Simplex Projection}

We now explain the various operations in the probability simplex
projection that do not map trivially to hardware implementation. We
will see in this section that the realization in hardware increases
the computation complexity (as measured by area usage) of the projections from
$\mathcal{O}\left(d\left(\log{d}\right)\right)$ to
$\mathcal{O}(d\left(\log{d}\right)^2)$. This increase is due to a parallel implementation of the sort operation.

\subsubsection{Sort}

As mentioned, sorting is a key aspect of the projection algorithm
developed in~\cite{draper_admm}. Sorting is relevant due to the
permutation symmetry of the parity polytope. Unfortunately, most
sorting methods are highly input dependent. For example,
quick-sort makes recursive calls of varying size depending
on the input data and the pivot value. We instead use sorting
networks. These are circuits composed of compare and swap modules that
perform the same operations regardless of the input data. Sorting
networks are explained in detail in~\cite{knuth_tacp}.

Our construction relies on delay-optimal sorting networks for dimension
16 or less, specified in~\cite{knuth_tacp}. For larger dimensions we
use Batcher's odd-even merge-sort (also discussed
in~\cite{knuth_tacp}) with recursive calls to the optimal sorting
networks when possible. Batcher's sort has the same recursive
structure as a standard merge sort. However, the merging operation is
performed in a clever manner that is input invariant. The price of
invariance is an increase in area to
$\mathcal{O}(d\left(\log{d}\right)^2)$ with a delay of
$\mathcal{O}(\left(\log{d}\right)^2)$. We comment that another sorting network, the AKS network, scales better in dimension than our
approach. However, according to~\cite{knuth_tacp}, AKS networks are completely impractical in any application.

\subsubsection{Prefix Sum}

Prefix sum refers to the sum operation on line~\ref{prefixSumAl} in
Algorithm~\ref{pspAlgorithm} whereby all partial sums of the sorted
vector are obtained. Performing this operation efficiently in serial
is straightforward. It is less obvious that a parallel algorithm can
perform this computation with logarithmic delay and linear
complexity. However, several such algorithms indeed do exist and have been
tested~\cite{vitoroulis_prefix} for their hardware compatibility in an
FPGA. In \cite{vitoroulis_prefix} the construction of Ladner and
Fischer~\cite{ladner_prefix} scales best with
respect to delay and area. We therefore use their method. It should
be noted that this method achieves the minimum possible delay of
$\ceil*{\log_2d}$ add operations on its critical path.

\subsubsection{Max Index}

Max index is the operation performed on line \ref{maxIndex} of
Algorithm \ref{pspAlgorithm}. The max index module takes as input a
bit vector of length $d$ where the high bits represent the set of
indices from which we want to select the maximum. The module
creates a one-hot bit vector where the high bit in the output
corresponds to the largest index high bit in the input. We can see
that in order for the $i^{\text{th}}$ output bit to be high, the
$i^{\text{th}}$ input bit must be high and all input bits above it
must be low. Writing out the Boolean expression for each output bit,
we can observe that this module is a prefix operation having almost
identical structure to the previously mentioned prefix sum. We take
the same approach as already described for prefix
sum. The difference being the use of logical AND operations instead of
additions.

Now that we have addressed how non-trivial operations will be mapped into hardware, we can evaluate the complexity of the projection hardware module. The sort operation is the most expensive module to implement. It gives the entire parity polytope projection module
an area complexity of
$\mathcal{O}(d\left(\log{d}\right)^2)$ with 
$\mathcal{O}(\left(\log{d}\right)^2)$ delay.

\section{Results}
\label{resultSec}

\subsection{Fixed Point Effects}
\subsubsection{Projections in the Unit Cube}
The objective of our first experiment is to determine the accuracy of the proposed projection algorithms when implemented in fixed-point. To test accuracy, we generate vectors uniformly at random from the three dimensional unit cube. We then project the vectors onto the parity polytope and probability simplex using varying degrees of fixed-point accuracy. 

To quantify the degradation with bit width, we start by generating double precision random vectors in MATLAB. We then apply a uniform quantizer to these vectors. This creates the fixed-point representation of the points to be projected. The fixed-point representation of the input and output vectors were both given 1 sign bit, 0 integer bits, and a varying number of fraction bits. This is because the absolute value of vector components will not exceed 1. These fixed-point vectors were then projected onto the parity polytope and probability simplex using the Verilog modules we created. As a point of comparison, the double precision input vectors were also projected onto the parity polytope and probability simplex. This latter projection was performed in MATLAB using double precision arithmetic. 

The average dimension normalized squared error is shown in Fig. \ref{unitCubeExp} for the quantized input, for the fixed-point parity polytope projection, and for the fixed-point probability simplex projection. The error measurement for each curve is relative to the corresponding double precision calculation. The fixed-point width is varied from 2 to 16 bits, corresponding to a number of fraction bits varying from 1 to 15.

\begin{figure}
	\centering
	\includegraphics[scale=0.6]{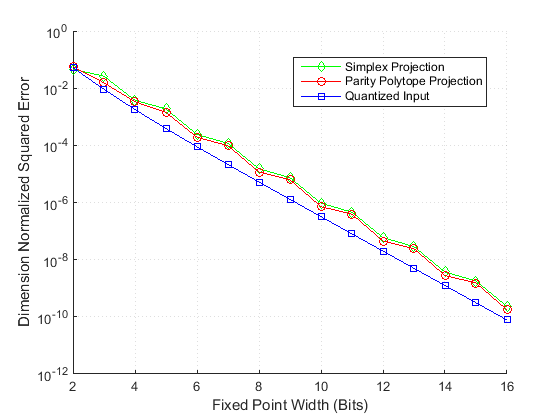}
	\caption{Accuracy for Unit Cube Projection onto $\mathbb{PP}_3$ and $\mathbb{S}_3$}
	\label{unitCubeExp}
\end{figure}

As would be expected for uniform quantization over a fixed range, the input error decreases exponentially (recall the 6dB/bit SNR rule of thumb). We also find that the projection error decreases exponentially in bit width. We believe the zig-zag in the projection error plot results from the fixed-point representation of $1/3$ in line \ref{prefixSumAl} of Algorithm \ref{pspAlgorithm}. Every second bit that is added to the binary representation of $1/3$ gives a disproportionate amount of accuracy since the representation is alternating $0$s and $1$s (i.e., $1/3 = 0.1010101\dots$).

The general shape of the curve - the exponential decrease in error - was seen for all trials on all dimensions (we simulated projections for various dimensions up to $d=27$). These results are not shown here since the zig-zagging effect is not as pronounced when there are more fractions from line \ref{prefixSumAl} of Algorithm \ref{pspAlgorithm} that cannot be represented exactly.

\subsubsection{Gaussian Projections}
The second fixed-point precision experiment investigates the trade-off between dynamic range and accuracy when projecting onto the parity polytope. Our approach is similar to the first experiment except the random input vector components have a normal distribution with mean zero and variance of 16.

Fig. \ref{gaussianExp} shows how the accuracy of various fixed-point representations compare to double precision projections. The input representation of all trials contains 1 sign sign bit. Each curve represents a different number of integer bits for the input representation. The x-axis shows the total fixed-point width. It is varied by changing the number of fraction bits in the input representation. Every output representation is composed of 1 sign bit, 1 integer bit, and a varying number of fraction bits giving the same fixed-point width as the input. We can see that the various error curves in Fig. \ref{gaussianExp} start at different locations. This is because each representation requires a different minimum width to accommodate its sign bit and integer bits.

\begin{figure}
	\centering
	\includegraphics[scale=0.6]{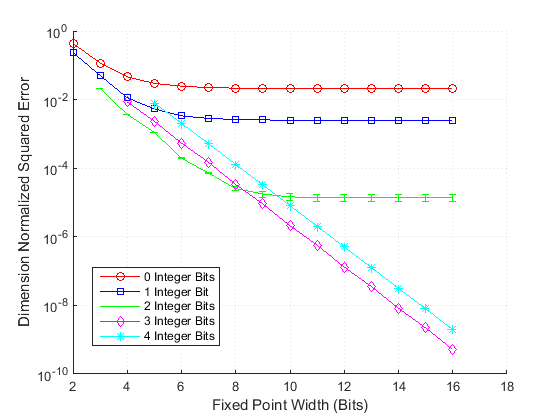}
	\caption{Accuracy for $\mathcal{N}(0,16)$ Projection onto $\mathbb{PP}_9$}
	\label{gaussianExp}
\end{figure}

For representations having 0, 1, or 2 integer bits in the input fixed-point representation, there is a small period of exponential decrease in error. This is followed by a saturation regime due to a lack of dynamic range. The 2 integer bit trial was the only curve with non-trivial error bars. Therefore they are displayed.

The 3 and 4 integer bit representations overcome the lack of dynamic range. This is indicated by the exponential decrease in error across the entire plotted range. We also observe that 3 integer bits performs better than 4 for all visible widths. This indicates that, for this class of input vectors, an extra bit is better spent on accuracy than on dynamic range.

Finally, we can also see what the optimal representation format is with regard to squared error. Below a fixed-point width of 8 bits, 2 integer bits is the best choice. Bits are better spent on accuracy than on dynamic range. However, when more than 8 bits are used in the representation, only having 2 integer bits results in a saturation regime, making the 3 integer bit representation the better choice.

\subsection{Resource Scaling}
We also synthesized the circuits in question for a range of dimensions to obtain real measurements of area and delay scaling. This synthesis was performed for an Altera Stratix V FPGA (model 5SGXEA7N2F45C2) which has 234,720 adaptive logic modules. The two projection circuits along with the sorting circuit were configured to have 8 bits of fixed-point precision: 1 sign bit, 1 integer bit, and 6 fraction bits for both input and output. The percentage of logic modules consumed by the synthesized circuits is a meaningful proxy measurement for area. In order to make the measurement even more meaningful, we disabled the synthesizer's ability to map arithmetic operations to dedicated DSP blocks.

\begin{figure}
	\centering
	\includegraphics[scale=0.6]{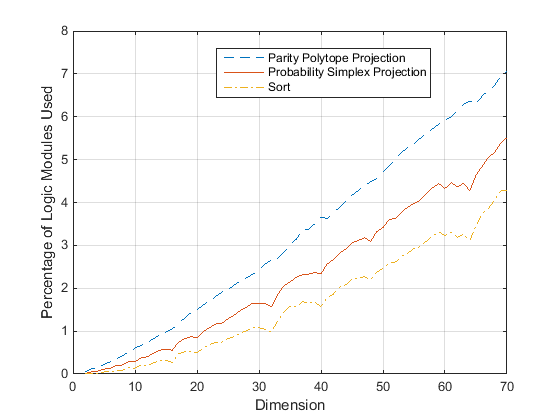}
	\caption{Measured Area Scaling}
	\label{areaGraph}
\end{figure}

The area measurements plotted in Fig. \ref{areaGraph} make the superlinear area scaling described in Section \ref{hardAdaptSec} visible. We can also see that all three modules follow similar area curves, indicating that the sort operation is the complexity dominating module. In order to observe practical delay scaling, we measured the critical path delay of the compiled circuits. Keep in mind that these circuits were compiled without any pipelining. Delay measurements are shown in Fig. \ref{delayGraph}.

\begin{figure}
	\centering
	\includegraphics[scale=0.6]{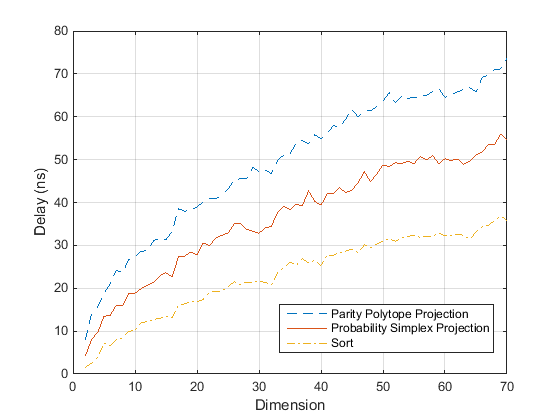}
	\caption{Measured Delay Scaling}
	\label{delayGraph}
\end{figure}

Similar to the area experiment, we see that delay measurements roughly follow the $\mathcal{O}(\left(\log{d}\right)^2)$ scaling predicted in Section \ref{hardAdaptSec}. In addition, the curves are quite similar. This is further evidence that the sort operation is the delay dominating module. Note also the jumps in area and delay at dimensions that are powers of two in Fig. \ref{areaGraph} and Fig. \ref{delayGraph}. These jumps occur because another layer must be added to the merge sort, emphasizing again the effect sorting has on the projection routines.

\section{Conclusion}
In this paper, we reviewed two essential optimization problems that have high throughput applications: projection onto the parity polytope and projection onto the probability simplex. Building off previous innovations, we were able to create hardware friendly algorithms to realize these projections. This enabled the creation of fixed-point implementations that were tested and found to have viable error curves and resource footprints.\balance 

The next step in our research is to apply these results to decoding error-correcting codes. We believe this will lead to a proof of concept for the LP decoding algorithm proposed in \cite{draper_admm}. We also believe that further research into more efficient, hardware compatible, projection algorithms would be extremely useful. This would be especially true if said research yields projection algorithms that completely avoid a sort operation.

\bibliographystyle{IEEEtran}
\bibliography{IEEEabrv,confPaper.bib}

\end{document}